\documentclass{LMCS}

\usepackage{enumerate}

\usepackage{amsmath}
\usepackage{amssymb}

\renewcommand{\emptyset}{\varnothing}
\newcommand{\N}{\mathbb{N}}
\newcommand{\Z}{\mathbb{Z}}
\newcommand{\Q}{\mathbb{Q}}
\newcommand{\R}{\mathbb{R}}
\def\vp{\varphi}

\def\doi{7 (2:5) 2011}
\lmcsheading%
{\doi}
{1--18}
{}
{}
{Aug.~\phantom07, 2010}
{May\phantom{.~0}5, 2011}
{}

\begin{document}

\title[Decidable Expansions of Labelled Linear Orderings]{Decidable Expansions of Labelled Linear Orderings}

\author[A.~B\`es]{Alexis B\`es\rsuper a}	
\address{{\lsuper a}University of Paris-Est Cr\'eteil, LACL}	
\email{bes@u-pec.fr}  

\author[A.~Rabinovich]{Alexander Rabinovich\rsuper{b}}	
\address{{\lsuper b}Tel-Aviv University, The Blavatnik School of Computer Science }	
\email{rabinoa@post.tau.ac.il}  
\thanks{{\lsuper{a,b}}This research was facilitated by the ESF project AutoMathA.
The second  author was  partially supported by ESF project Games and  EPSRC grant.}	



\titlecomment{A preliminary version of this paper appeared in \cite{BesRabinovich10}.}


\begin{abstract}
  \noindent Consider a linear ordering equipped with a finite sequence
  of monadic predicates.  If the ordering contains an interval of
  order type $\omega$ or $-\omega$, and the monadic second-order
  theory of the combined structure is decidable, there exists a
  non-trivial expansion by a further monadic predicate that is still
  decidable.
\end{abstract}

\maketitle

\section*{Introduction}\label{S:one}

In this paper we address definability and decidability issues for monadic second order (shortly: MSO) theories of
 labelled linear orderings. Elgot and Rabin ask in \cite{ElgotRabin66} whether there exist maximal decidable structures, i.e.,  structures $M$ with a decidable first-order (shortly: FO) theory and such that the FO theory of any expansion of $M$ by a non-definable predicate is undecidable. This question is still open. Let us mention some partial results:

\begin{enumerate}[$\bullet$]
\item Soprunov proved in \cite{Soprunov88} that every structure in which a regular ordering is interpretable is not maximal. A partial ordering $(B,<)$ is said to be regular if for every $a \in B$ there exist distinct elements $b_1,b_2 \in B$ such that $b_1<a$, $b_2<a$, and no element $c \in B$ satisfies both $c<b_1$ and $c<b_2$. As a corollary he also proved that there is no maximal decidable countable structure if we replace FO by weak MSO logic.
\item In \cite{BesCegielski:2009:JMS}, B\`es and C\'egielski consider a weakening of the Elgot-Rabin question, namely the question of whether all structures ${M}$ whose FO theory is decidable can be expanded by some constant in such a way that the resulting structure still has a decidable theory. They answer this question negatively by proving that there exists a structure ${M}$ with a decidable MSO theory and such that any expansion of ${M}$ by a constant has an undecidable FO theory.
\item The paper \cite{BC08} gives  a sufficient condition in terms of the Gaifman graph of $M$ which ensures that ${M}$ is not maximal. The condition is the following: for every natural number $r$ and every finite set $X$ of elements of the base set $|M|$ of $M$ there exists an element $x \in |M|$ such that the Gaifman distance between $x$ and every element of $X$ is greater than $r$. 
\end{enumerate}
We investigate the Elgot-Rabin problem for the class
of labelled linear orderings, i.e.,  infinite structures
$M=(A;<,P_1,\dots,P_n)$ where $<$ is a linear ordering over $A$ and
the $P_i$'s denote unary predicates. This class is interesting with
respect to the above results, since on one hand no regular ordering
seems to be FO interpretable in such structures, and on the other
hand their associated Gaifman distance is trivial, thus they do not
satisfy the criterion given in \cite{BC08}.

In this paper we focus on MSO logic rather than FO. The main result
of the paper is that for every labelled linear ordering $M$ such
that $(A,<)$ contains an interval of order type $\omega$ or
$-\omega$ and the MSO theory of $M$ is decidable,
 there exists an expansion $M'$ of $M$ by a monadic  predicate
which is not MSO-definable in $M$, and such that the MSO theory of
$M'$   is still decidable. Hence, $M$ is not
maximal.  The result holds in particular when $(A,<)$ is order-isomorphic to the order of the naturals
 $\omega=(\N,<)$,  or to  the order  $\zeta=(\Z,<)$ of the integers, or to any infinite ordinal, or more generally any
infinite scattered ordering (recall that an ordering is scattered if
it does not contain any dense sub-ordering).

The structure of the proof is the following: we first show that the result holds for $\omega$ and $\zeta$.
 For the general case, starting from $M$, we use some definable equivalence relation on $A$ to cut $A$ into intervals whose
  order type is either finite, or of the form $-\omega$, $\omega$, or $\zeta $. We then define the new predicate on each interval (using the constructions given for $\omega$
  and $\zeta$), from which we get the definition of $M'$. The reduction from $MSO(M')$ to $MSO(M)$ uses Shelah's composition theorem, which allows us to reduce the MSO theory of an ordered sum of structures to the MSO theories of the summands.

 The main reason to consider MSO logic rather than FO is that it actually simplifies the task. Nevertheless we discuss some partial results and perspectives for FO logic in the conclusion of the paper.

Let us recall some important decidability results for MSO theories of linear orderings (the case of labelled linear orderings will be discussed later for $\omega$ and $\zeta$). In his seminal paper \cite{Buchi62}, B\"uchi proved that languages of $\omega-$words recognizable by automata coincide with languages
definable in the MSO theory of $\omega$, from which he deduced decidability of the theory. The result (and the automata method) was then extended to the MSO theory of any  countable ordinal \cite{Buchi65}, to $\omega_1$, and to any ordinal less than $\omega_2$ \cite{BuchiZ83}. Gurevich, Magidor and Shelah prove \cite{GurevichMS83} that decidability of MSO theory of $\omega_2$ is independent of ZFC. Let us mention results for linear orderings beyond ordinals. Using automata, Rabin \cite{Rabin69} proved
decidability of the MSO theory of the binary tree, from
which he deduces decidability of the MSO theory of $\Q$,
which in turn implies decidability of the MSO theory of the class of
countable linear orderings. Shelah \cite{Shelah75} improved model-theoretical techniques that allowed him to
reprove almost all known decidability results about MSO
theories, as well as new decidability results for the case of linear
orderings, and in particular dense orderings. He proved in particular that the MSO theory
of $\R$ is undecidable. The frontier between decidable and undecidable cases was specified in later papers by Gurevich and Shelah \cite{Gurevich79,GurevichS79,GurevichS83}; we refer the reader to the survey \cite{Gurevich85}.

 Our result is also clearly related to the problem of building larger and larger classes of structures with a decidable MSO theory. For an overview of recent results in this area see \cite{BlumensathColcombetLoeding07,Thomas08}.

\section{Preliminaries}

\subsection{Labelled Linear Orderings }

We first recall useful definitions and results about linear
orderings. A good reference on the subject is Rosenstein's book
\cite{Rosenstein82}.

A \emph{linear ordering}~$J$ is a total ordering. We denote by $\omega$ (respectively $\zeta$) the order type of $\N$ (respectively $\Z$). Given a linear
ordering~$J$, we denote by $-J$ the \emph{backwards} linear ordering
obtained by reversing the ordering relation.

Given a linear ordering $J$ and $j \in J$, we denote by $[j]$ the interval $[j,j]$. 
An ordering is \emph{dense} if it contains no pair of consecutive
elements. An ordering is \emph{scattered} if it contains no dense sub-ordering.

In this paper we consider {\em labelled} linear orderings, i.e.,
linear orderings  $(A,<)$ equipped with a function $f:A \to \Sigma$
where $\Sigma$ is a finite nonempty set.

\subsection{Logic}

Let us briefly recall useful elements of monadic second-order logic, and settle some notations. For more details about MSO logic see e.g. \cite{Gurevich85,Thomas97a}. Monadic second-order logic is an extension of first-order logic  that allows to quantify over elements as well as subsets of the domain of the structure. Given a signature $L$, one can define the set of (MSO) formulas over ${L}$ as well-formed formulas that can use first-order variable symbols $x,y,\dots$ interpreted as elements of the domain of the structure, monadic second-order variable symbols $X,Y,\dots$ interpreted as subsets of the domain, symbols from ${L}$, and a new binary  predicate $x \in X$ interpreted as ``$x$ belongs to $X$". A sentence is a formula without free variable. As usual, we  often confuse logical symbols with their interpretation.  Given a signature ${L}$ and an  ${L}-$structure $M$ with domain $D$, we say that a relation $R \subseteq D^m \times (2^D)^n$ is {(MSO) definable} in $M$ if and only if there exists a formula over ${L}$, say $\varphi(x_1,\dots,x_m,X_1,\dots,X_n)$, which is true in $M$ if and only if $(x_1,\dots,x_m,X_1,\dots,X_n)$ is interpreted by an $(m+n)-$tuple of $R$.  Given a structure $M$ we denote by $MSO(M)$ (respectively $FO(M)$) the monadic second-order (respectively first-order) theory of $M$. We say that $M$ is maximal if $MSO(M)$ is decidable and $MSO(M')$ is undecidable for every expansion $M'$ of $M$ by a predicate which is not definable in $M$.

We can identify labelled linear orderings with structures of the form $M=(A,<,P_1,\dots,P_n)$ where $<$ is a binary relation interpreted as a linear ordering over $A$, and the $P_i$'s denote unary predicates.  We use the notation $\overline{P}$ as a shortcut for
the $n$-tuple $(P_1,\dots, P_n)$.  The structure $M$ can be seen as a word indexed by $A$ and over the alphabet $\Sigma_n=\{0,1\}^n$; this word will be denoted by $w(M)$. For every interval $I$ of $A$ we denote by $M_I$ the sub-structure of $M$ with domain $I$.

Let $\Sigma$ and $\Sigma'$ be relational  signatures, $M$ a
$\Sigma$-structure with domain $A$ and $M'$ a $\Sigma'$-structure with domain $A'$. We say that $M$ is (MSO) interpretable in $M'$ if there exist a subset $D$ of $A'$ and a surjective map ${\mathcal I}:D \to A$ such that:
\begin{enumerate}[$\bullet$]
\item $D$ is MSO definable in $M'$;
\item The equivalence relation $EQ_{\mathcal I}=\{(x,y) \in A': {\mathcal I}(x)={\mathcal I}(y) \}$ is MSO definable in $M'$;
\item For every $m$-ary symbol $R$ of $\Sigma$, there exists a MSO $\Sigma'-$formula $\varphi_R$  such that
$$M\models R({\mathcal I}(a_1), \dots ,{\mathcal I}(a_{m})) \Leftrightarrow
M'\models \varphi_R(a_1, \dots,a_{m})$$
for all
$a_1,\dots, a_{m}\in D$.
\end{enumerate}

The following property of
interpretations is well-known.

\begin{lem}\label{lem:interp}
If $M$ is interpretable in $M'$ then $MSO(M)$ is recursive in $MSO(M')$.
\end{lem}

\subsection{Elements of the Composition Method}\label{subsec:composition}

In this paper we rely heavily on composition methods, which allow us to compute the theory of a sum of structures from the ones of its summands.
 For an overview of the subject see \cite{BlumensathColcombetLoeding07,Thomas97,Mak04}. In this section we recall useful definitions and results.

  The quantifier depth of  a formula  $\varphi$ is denoted by $qd(\varphi)$.
  Let $n\in \N$,  $\Delta$   any finite signature  that contains only relational
symbols, and $M_1, M_2$ be $\Delta$-structures. We say that $M_1$
and $M_2$ are $n$-\emph{equivalent}, denoted $M_1 \equiv^n M_2$, if
for every sentence $\varphi $ of quantifier depth at most $n$,  $M_1 \models
\varphi$ iff $M_2 \models \varphi$.

Clearly, $\equiv^n$ is an equivalence relation. For any $n \in \N$
and $\Delta$, the set of sentences of quantifier depth $\leq n$ is
infinite. However, it contains only finitely many semantically
distinct sentences, so there are only finitely many
$\equiv^n$-classes of $\Delta$-structures. In fact, we can compute
representatives for these classes.

\begin{lem}[Hintikka Lemma]\label{comp}

For each $n\in \N$ and a  finite signature  $\Delta$ that contains
only relational symbols,
 we can compute a \emph{finite} set
$H_{n}(\Delta) $ of $\Delta$-sentences of quantifier depth at most
$n$
 such that:
\begin{enumerate}[$\bullet$]
\item If $\tau_1,\tau_2 \in H_{n}(\Delta) $ and $\tau_1\neq \tau_2$,
then $\tau_1\wedge\tau_2$ is unsatisfiable.
\item If $\tau \in H_{n}(\Delta) $ and $qd(\varphi )\leq n$, then either
  $\tau \rightarrow \varphi$ or $\tau \rightarrow \neg\varphi$.
  Furthermore, there is an algorithm that, given such $\tau$ and $\varphi$,
  decides which of these two possibilities holds.
\item For every $\Delta$-structure  $M$ there is a
  \emph{unique} $\tau \in H_{n}(\Delta) $ such that $M \models \tau$.
\end{enumerate}
Elements of $H_{n}(\Delta)$ are called $(n,\Delta)$-\emph{Hintikka
sentences}.

\end{lem}
Given a $\Delta$-structure  $M$ we denote by $T^n(M)$ the unique
element of $H_{n}(\Delta) $ satisfied in $M$ and call it the
$n$-\emph{type} of $M$. Thus, $T^n(M)$ determines (effectively)
which sentences  of quantifier-depth $\leq n$ are satisfied in $M$.

As a simple consequence, note that the MSO theory of a structure $M$
is decidable if and only if the function $k \mapsto T^k(M)$ is
recursive.

The sum of structures corresponds to concatenation; let us give a general definition.

\begin{defi}[sum of chains]
Consider an index structure $Ind=(I,<^I)$ where $<^I$ is a linear ordering.
 Consider a  signature  $\Delta=\{<,{P_1, \dots ,P_l}\}$, where $P_i$ are unary predicate names,   and a family $(M_i)_{i \in I}$ of $\Delta$-structures
  $M_i=(A_i;<^i,{P_1}^i,\dots ,P_l^i)$
 with disjoint domains and such that the interpretation $<^i$ of $<$ in each $M_i$ is
  a linear ordering. We define the {\em ordered sum} of the family $(M_i)_{i \in I}$ as the $\Delta-$structure $M=(A; <^M, {P_1}^M,\dots, P_l^M)$ where
\begin{enumerate}[$\bullet$]
\item $A$ equals the union of the $A_i$'s
\item $ x<^M y$ holds if and only if ($x \in A_i$ and $y \in A_j$ for some $i<^I j$), or ($x,y \in A_i$ and $x<^i y$)
\item for every $x \in A$ and every $k \in \{1,\dots,l\}$, ${P}_k^M(x)$ holds  if and only if $M_j \models {P}_k^j(x)$ where $j$ is such that $x \in A_j$.
\end{enumerate}
If the domains of the $M_i$  are not disjoint, replace them
with isomorphic chains that have disjoint domains, and proceed
as before.

We shall use the notation $M=\sum_{i \in I} M_i$ for  the {ordered sum} of the family $(M_i)_{i \in I}$ .

If $I=\{1,2\}$ has   two elements, we  denote  $\sum_{i \in I} M_i$
by $M_1+M_2$.
\end{defi}

We need the following composition theorem on ordered sums (see e.g. \cite{Thomas97}):
\begin{thm}\label{thm:simplecomposition}
\item[(a)] The $k$-types of labelled linear orderings
$M_0, M_1$  determine the $k$-type of the ordered sum $M_0
+ M_1$, which moreover can be computed from the $k$-types
of $M_0$ and $M_1$.

\item[(b)] If the labelled linear orderings $M_0, M_1, \ldots$ all have the
same $k$-type, then this $k$-type determines the $k$-type
of $\Sigma_{i \in \N} M_i$, which moreover can be computed from
the $k$-type of $M_0$.
\end{thm}
Part (a) of the theorem justifies the notation $s + t$ for
the $k$-type of a linear ordering which is the sum of two linear orderings
of $k$-types $s$ and $t$, respectively. Similarly, we write
$t \times \omega$ for the $k$-type of a sum $\Sigma_{i \in \N} M_i$ where
all $M_i$ have $k$-type $t$.

For every linear ordering $(A,<)$ and every $x \in A$ let $I_{<x}=\{y \in A: \ y<x\}$ and $I_{\geq x}=\{y \in A: \ y \geq x\}$.  
The following is a well-known consequence of Theorem \ref{thm:simplecomposition}(a) (see e.g \cite[Theorem A.1]{Gabbay89}).

\begin{cor}\label{cor:def-elements}
Let $M=(A,<,\overline{P})$, $k \geq 1$, and $a,b \in A$ be such that $I_{<a}$ and $I_{<b}$ are nonempty sets. Assume that 
$T^{k}(M_{I_{<a}})=T^k(M_{I_{<b}})$ 
and $T^k(M_{I_{\geq a}})=T^k(M_{I_{\geq b}})$. Then for every formula $\varphi(x)$ such that $qd(\varphi) < k$, we have 
$$M\models \vp(a) \textrm{ if and only if } M\models \vp(b).$$
\end{cor}

We shall use the following result.

\begin{prop}\label{prop:def-interval}
Let  $M = (A, <, \overline{P})$, $I$ be an interval of $A$, and $Q \subseteq A$. If $Q$ is definable in $M$ then $Q \cap I$ is definable in $M_I$.
\end{prop}

\begin{proof}
Assume that the formula $\varphi(x)$ defines $Q$ in $M$. Note that $Q$ is the unique predicate which satisfies the formula $\psi(X)  \equiv \forall x (x \in X \leftrightarrow \varphi(x))$. Let $m=qd(\psi)$. For every $R \subseteq A$, let us denote by $M^R$ the expansion of $M$ with a new monadic predicate $\mathbf X$ interpreted by $R$. By our hypothesis, we have $T^m(M^R)=T^m(M^Q)$ if and only if $Q=R$. 

Let $J_1$ (respectively $J_2$) be the set of elements of $A$ less than (respectively greater than) every element of $I$. We have $M^R=M^R_{J_1}+M^R_I+M^R_{J_2}$. 
We claim that $T^m(M^R_I)=T^m(M^Q_I)$ if and only if $R \cap I= Q \cap I$. Indeed assume for a contradiction that there exists $R$ such that $T^m(M^R_I)=T^m(M^Q_I)$ and $R \cap I \ne Q \cap I$. Consider $Q' \subseteq A$ such that $Q' \cap J_1= Q \cap J_1$, $Q' \cap J_2= Q \cap J_2$ and $Q' \cap I= R \cap I$. On the one hand we have $Q' \ne Q$, and on the other hand by Theorem \ref{thm:simplecomposition}(a) we have $$T^m(M^{Q'})=T^m(M^Q_{J_1})+T^m(M^R_I)+T^m(M^Q_{J_2})=T^m(M^Q_{J_1})+T^m(M^Q_I)+T^m(M^Q_{J_2})=T^m(M^Q)$$  which contradicts the fact that $Q$ is the unique predicate
which satisfies $\psi(X)$ in $M$.

The type $T^m(M^Q_I)$ is a sentence $\theta$ over the signature $\{<, \overline{P}\} \cup \{{\mathbf X}\}$. We can see it as a formula $\theta(X)$ with one free monadic variable $X$ over the signature $\{<, \overline{P}\}$. We have $M_I \models \theta(X)$ if and only if $X=Q \cap I$. Thus the formula $\gamma(x)$ defined as $\exists X(\theta(X) \wedge x \in X)$ defines $Q \cap I$ in $M_I$.
\end{proof}

\subsection{Decomposition of a labelled linear ordering}
Let $M=(A,<,\overline{P})$ be a labelled linear ordering and let $\sim$ be an equivalence relation on $A$. If the $\sim-$equivalence classes are intervals of $A$ we say that $\sim$ is a \emph{convex} equivalence relation. In
this case the set of $\sim$-equivalence classes can be naturally
ordered by $i_1\leq i_2$ iff $\exists x_1\in i_1 \ \exists x_2\in
i_2 \ (x_1\leq x_2)$. We denote by $M/_\sim$ the linear order of
$\sim$-equivalence classes. The mapping that assigns to every $x\in
A$ its $\sim$-equivalence class is said to be\emph{ canonical}.

Let  $\sim$ be  a convex equivalence relation on $A$. Then
$M=\sum_{i\in M/_\sim} M_i$, where $M_i$ is the substructure of $M$ with domain 
the equivalence class $i$.
\begin{lem} \label{lem:decomp} If $\sim$ is a convex equivalence relation which is
definable in $M$, then
\begin{enumerate}[\em(1)]
\item  $M/_\sim$ is interpretable in $M$.
\item Let
$\vp_1, \dots \vp_k$ be sentences in the signature of $M$. Let 
$$N=(M/_\sim,<,Q_{\vp_1}, \dots, Q_{\vp_k})$$  where  $Q_{\vp_l}=\{ i\in
M/_\sim\mid M_i\models \vp_l\}$ for every $l$. Then $N$ is interpretable in $M$.
\end{enumerate}
\end{lem}

\section{The Case of $\N$}

In this section we prove that there is no maximal structure of the form $(\N,<,\overline{P})$ with respect to MSO logic.  The proof is based upon
 results from  \cite{Rabinovich07a} . Let us first briefly review results related to the decidability of the MSO theory of expansions of $(\N,<)$.
 B\"uchi \cite{Buchi62} proved decidability of $MSO(\N,<)$ using automata. On the other hand it is known that $MSO(\N,+)$, and even $MSO(\N,<,x \mapsto 2x)$, are undecidable \cite{Robinson58}. Elgot and Rabin study in \cite{ElgotRabin66} the MSO theory of structures of the form $(\N,<,P)$, where $P$ is some unary predicate. They give a sufficient condition on $P$ which ensures decidability of the MSO theory of $(\N,<,P)$. In particular the condition holds when $P$ denotes the set of factorials, or the set of powers of any fixed integer. The frontier between decidability and undecidability of related theories was explored in numerous later papers \cite{CartonThomas02,Fratani09,Semenov83,Semenov84,RabinovichT06,Rabinovich07a,Siefkes70,Thomas75}. Let us also mention that \cite{Semenov83} proves the existence of unary
 predicates $P$ and $Q$ such that both $MSO(\N,<,P)$ and $MSO(\N,<,Q)$ are decidable while $MSO(\N,<,P,Q)$ is undecidable.

Most decidability proofs for $MSO(\N,<,\overline{P})$ are related somehow to the possibility of cutting $\N$ into segments whose $k-$type is ultimately constant,
from which one can compute the $k-$type of the whole structure (using Theorem \ref{thm:simplecomposition}). This connection was specified in
\cite{Rabinovich07a} (see also \cite{RabinovichT06}) using the notion of homogeneous sets.

\begin{defi}[$k$-homogeneous set]
Let $k \geq 0$. A set $H = \{h_0 < h_1 < \ldots\} \subseteq \N$
is called {\em $k$-homogeneous for $M = (\N, <, \overline{P})$}, if all sub-structures $M_{[h_i, h_j[}$ for $i < j$ (and
hence all sub-structures $M_{[h_i, h_{i+1}[}$ for $i \geq 0$) have the same
$k$-type.
\end{defi}
This notion can be refined as follows.

\begin{defi}[uniformly homogeneous set]
A set $H = \{h_0 < h_1 < \ldots \} \subseteq \N$ is called
{\em uniformly homogeneous
for $M = (\N, <, \overline{P})$}
if for each $k$ the set
$H_k = \{h_k < h_{k+1} < \ldots \}$ is $k$-homogeneous.
\end{defi}
The following result \cite{Rabinovich07a} settles a tight connection between $MSO(\N,<,\overline{P})$ and uniformly homogeneous sets.
\begin{thm}(\cite{Rabinovich07a})\label{th:rechomog}
For every  $M = (\N, <, \overline{P})$, the $MSO$ theory of $M$ is decidable
if and only if (the sets $\overline{P}$ are recursive and there exists a recursive uniformly homogeneous set for $M$).
\end{thm}

The proof of Theorem \ref{th:rechomog} given in  \cite{Rabinovich07a} actually shows the following. 

\begin{thm}(\cite{Rabinovich07a})\label{th:rechomog-restated}
For every  $M = (\N, <, \overline{P})$, there exists $H$ which is uniformly homogeneous for $M$ and recursive in $MSO(M)$. 
\end{thm}

We shall prove that any set $H$ which is uniformly homogeneous for $M$ and recursive in $MSO(M)$ can be used to expand $M$ and get non-maximality. 

\begin{lem}\label{lem:undefN}
For every  $M = (\N, <, \overline{P})$, if $H$ is uniformly homogeneous for $M$ then $H$ is not definable in $M$.
\end{lem}

\begin{proof}
Let $H$ be uniformly homogeneous for $M$, and let $h_0<h_1<...$ be the sequence of elements of $H$. Assume for a contradiction that $H$ is definable in $M$. Then it follows that the set $H_{even}=\{h_{2i}: i \in \N \}$ is also definable in $M$. We shall prove that for every $k \geq 2$, the elements $a_k=h_{2k}$ and $b_k=h_{2k+1}$ satisfy in $M$ the same formulas $\varphi(x)$ with quantifier depth $k$. Since $a_k \in H_{even}$ and $b_k \not\in H_{even}$, it follows that $H_{even}$ is not definable in $M$ by any formula $\varphi(x)$ with quantifier depth less than $k$, from which we  get a contradiction.

For every $i\geq 0$ let us denote by $\tau_i$ the $i-$type of $M_{[h_i,h_{i+1}[}$. Let $k \geq 2$. 
 We have $a_{k-1}=h_{2(k-1)}$ and $2(k-1) \geq k$, therefore by definition of $H$ we have
$$ T^k(M_{[a_{k-1},a_{k}[}) =  \tau_k = T^k(M_{[a_{k-1},b_{k}[}) $$
thus
\begin{eqnarray}\label{eq:left}
T^k(M_{[0,b_k[}) & = &T^k(M_{[0,a_{k-1}[})+ T^k(M_{[a_{k-1},b_{k}[}) \nonumber \\
 & = & T^k(M_{[0,a_{k-1}[})+ T^k(M_{[a_{k-1},a_{k}[}) \nonumber \\
 & = & T^k(M_{[0,a_k[}). 
\end{eqnarray}
  On the other hand by definition of $H$ we also have 
\begin{eqnarray}\label{eq:right}
T^k(M_{[a_k,\infty[})=\tau_k \times \omega=T^k(M_{[b_k,\infty[}).
\end{eqnarray}

It follows from  Equations (\ref{eq:left}) 
and (\ref{eq:right}), and Corollary \ref{cor:def-elements}   that $a_k$ and $b_k$ satisfy in $M$ the same formulas $\varphi(x)$ with quantifier depth less than $k$.
\end{proof}

\begin{lem}\label{lem:recHP}
Let  $H $  be a uniformly homogeneous set for $M = (\N, <,
\overline{P})$.  Then
  the  $MSO$ theory of $M' = (\N, <, H, \overline{P})$ is recursive in
  $ (H, \overline{P})$.
\end{lem}

\begin{proof}
Let us denote by $h_0<h_1<...$ the sequence of elements of $H$.
We have
$$T^k(M')=T^k(M'_{[0,h_{k+2}[})+\sum_{i\geq k+2} T^k(M'_{[h_{i},h_{i+1}[})$$
For every $i \geq k+2$, the only element of $H$ in the interval $[h_{i},h_{i+1}[$ is $h_{i}$. This implies that for every MSO-sentence $\varphi$ in the signature of $M'$ such that $qd(\varphi)=k$, we have
$$M'_{[h_{i},h_{i+1}[} \models \varphi$$
if and only if 
\begin{eqnarray}\label{formulareduction}
M_{[h_{i},h_{i+1}[} \models \exists x ((\forall y \ x \leq y) \wedge \varphi^*)
\end{eqnarray}
where $x$ is any variable which does not appear in $\varphi$, and $\varphi^*$ is obtained  by replacing in $\varphi$ every atomic formula of the form $H(z)$ (where $z$ denotes any first-order variable) by $z=x$.

The formula in (\ref{formulareduction}) has quantifier depth $\leq k+2$. This implies that the $k-$type of $M'_{[h_{i},h_{i+1}[}$ can be computed from the $(k+2)-$type of $M_{[h_{i},h_{i+1}[}$. Since $H$ is uniformly homogeneous there is a $(k+2)-$type $\tau_{k+2}$  such that   the $(k+2)-$type of $M_{[h_{i},h_{i+1}[}$ equals $\tau_{k+2}$ for every $i \geq k+2$. Moreover, from $H$ and $\overline P$ we can compute $h_{k+2}$ and $h_{k+3}$, and then the $(k+2)-$type of $M_{[h_{k+2},h_{k+3}[}$, which equals $\tau_{k+2}$.  Thus one can compute $\tau'$ such that the $k-$type of $M'_{[h_{i},h_{i+1}[}$ equals $\tau'$ for every $i \geq k+2$. Therefore we have
$$T^k(M')=T^k(M'_{[0,h_{k+2}[})+\sum_{i\geq k+2} \tau'=T^k(M'_{[0,h_{k+2}[})+ \tau' \times \omega$$
Moreover one checks that $T^k(M'_{[0,h_{k+2}[})$ is computable from $H$ and $\overline P$.

 Finally $T^k(M')$ is computable by Theorem \ref{thm:simplecomposition}. 
\end{proof}

The previous results allow us to show that no structure $M=(\N, <, \overline{P})$ is maximal.

\begin{prop}\label{prop:1stConstructionN}
For every structure $M=(\N,<,\overline{P})$ there exists an expansion
$M'$ of $M$ by a predicate $P_{n+1}$ such that $P_{n+1}$ is not definable in $M$ and $MSO(M')$ is recursive in $MSO(M)$. In particular, if $MSO(M)$ is decidable,  then  $MSO(M')$ is decidable.
\end{prop}

\begin{proof}
By Theorem \ref{th:rechomog-restated}, there exists a set $H$ which is uniformly homogeneous for $M$ and recursive in $MSO(M)$. We set $P_{n+1}=H$. By 
Lemma \ref{lem:undefN}, $H$ is not definable in $M$. By Lemma \ref{lem:recHP}, $MSO(M')$ is recursive in $H$ and $\overline{P}$, which are both recursive in $MSO(M)$. 
\end{proof}

In the proof of the general result (see Sect. \ref{sec:general}), we start from a labelled linear ordering $M=(A,<,\overline{P})$  and try to define an expansion $M'$ such that $MSO(M')$  is recursive in $MSO(M)$. In some case the expansion $M'$ of $M$ will be defined by applying the above construction to infinitely many intervals of $A$ of order type $\omega$. In order to get a reduction from $MSO(M')$ to $MSO(M)$, we need that the reduction algorithm for such intervals is uniform. This leads to the following Proposition, which can be seen as a uniform version of Proposition \ref{prop:1stConstructionN}.

\begin{prop}\label{prop:2ndConstructionN}
There exists a function $E$ and two recursive functions $g_1,g_2$ such that $E$ maps every structure $M=(\N,<,\overline{P})$ to an expansion
$M'$ of $M$ by a predicate $P_{n+1}$ such that
\begin{enumerate}[\em(1)]
\item $P_{n+1}$ is not definable in $M$;
\item $g_1$ computes $T^{k}(M')$ from ${k}$ and $T^{g_2(k)}(M)$.
\end{enumerate}
Hence $MSO(M')$ is recursive in $MSO(M)$. In particular, if $MSO(M)$ is decidable,  then  $MSO(M')$ is decidable.
\end{prop}

The proof of the above proposition relies on the construction of a special uniformly homogeneous set $H$ for $M$.

\begin{prop}(\cite{Rabinovich07a}) \label{prop:unifhomog}
There exists a recursive function which maps every $k \geq 0$ to some formula $\varphi_k(X)$ such that 
\begin{enumerate}[\em(1)]
\item for every structure $M=(\N,<,\overline{P})$ there exists a unique $X \subseteq \N$ such that $\varphi_k(X)$ holds in $M$. This set will be denoted by $H_k$.
\item For every $k\geq 0$, the sets $H_k$ and $H_{k+1}$ have the same $k$ first elements.
\item The set $H=\{h_k: h_k \textrm{ is the $k-$th element of $H_k$} \}$ is uniformly homogeneous.
\end{enumerate}
\end{prop}

In the above Proposition, for every $i \in \N$ the $i-$th element $h_i$ of $H$ is the $i-$th element of the unique set $X \subseteq \N$ such that $\varphi_i(X)$ holds in $M$. This implies that for every $i \in \N$, one can compute a formula $\alpha_i(x)$ which defines $h_i$ in every $M$.

\begin{cor}(\cite{Rabinovich07a})\label{cor:compH}
There exist   recursive functions $q$  and $f$ such that for every
structure $M=(\N,<,\overline{P})$, and corresponding $H$ as in
Proposition \ref{prop:unifhomog}, $f$ computes the $i$-th element of
$H$ from $i$ and  $T^{q(i)}(M)$.
\end{cor}

\noindent{\bf Proof of Proposition \ref{prop:2ndConstructionN}}.
Let $M=(\N,<,\overline{P})$. Let $P_{n+1}=H$ where $H$ is the set associated to $M$ as stated in Proposition \ref{prop:unifhomog}, and let 
 $M'$ be the expansion of $M$ by $P_{n+1}$. Let $E$ be the function which maps $M$ to $M'$. By Lemma \ref{lem:undefN}, $P_{n+1}$ is not MSO definable in $M$. 
 
In order to show that there exist recursive functions $g_1,g_2$ such that $g_1$ computes $T^{k}(M')$ from ${k}$ and $T^{g_2(k)}(M)$, we have to revise the proof of Lemma \ref{lem:recHP}. In this proof the reduction from $T^k(M')$ to $MSO(M)$ comes from the fact that the types $\tau_{k+2}=T^{k+2}(M_{[h_{k+2},h_{k+3}[})$, and $T^k(M'_{[0,h_{k+2}[})$, are computable from $H$ and $\overline P$. We shall prove that for our specific choice of $H$, these  types are computable from $T^{g_2(k)}(M)$ for some recursive function $g_2$. 

First, let us consider the type $T^{k+2}(M_{[h_{k+2},h_{k+3}[})$. For every sentence $\varphi$ in the signature of $M$ such that $qd(\varphi)=k+2$, we have
\begin{equation}
M_{[h_{k+2},h_{k+3}[} \models \varphi \label{eqtransl0}
\end{equation}
if and only if
\begin{equation}
M \models \exists x \exists y (\alpha_{k+2}(x) \wedge \alpha_{k+3}(y) \wedge \varphi')\label{eqtransl1}
\end{equation}
where $x$ and $y$ are variables which do not appear in $\varphi$, and $\varphi'$ is obtained from $\varphi$ by relativizing all quantifiers to the interval $[x,y[$. The formula in (\ref{eqtransl1}) has quantifier depth $$q_1(k)=2+\max (qd(\alpha_{k+2}),  qd(\alpha_{k+3}), k+2)$$ since $qd(\varphi')=qd(\varphi)$. The equivalence between (\ref{eqtransl0}) and (\ref{eqtransl1}) implies that $\tau_{k+2}$ is computable from $T^{q_1(k)}(M)$.

Consider now the type $T^k(M'_{[0,h_{k+2}[})$. 
For every sentence $\varphi$ in the signature of $M'$ such that $qd(\varphi)=k$, we have 
\begin{equation}
M'_{[0,h_{k+2}[} \models \varphi \label{eqtransl2}
\end{equation}
iff
\begin{equation}
M \models \exists x  (\alpha_{k+2}(x) \wedge  \varphi'')\label{eqtransl3}
\end{equation}
where $x$ is a variable which does not appear in $\varphi$, and $\varphi''$ is obtained from $\varphi$ by relativizing all quantifiers to the interval $[0,x[$, and replacing (in the resulting formula) every atomic formula of the form $H(z)$ by the formula $\bigvee_{0 \leq i \leq k+1} \alpha_i(z)$. The formula in (\ref{eqtransl3}) has quantifier depth 
$$q_2(k)=1+\max (qd(\alpha_{k+2}), k+\max(qd(\alpha_i): 0 \leq i \leq k+1) ).$$ The equivalence between (\ref{eqtransl2}) and (\ref{eqtransl3}) implies that $T^k(M'_{[0,h_{k+2}[})$ is computable from $T^{q_2(k)}(M)$.

Finally we obtain that $\tau_{k+2}$ and $T^k(M'_{[0,h_{k+2}[})$, and thus $T^k(M')$, are computable from   $T^{g_2(k)}(M)$ where $g_2=\max(q_1,q_2)$ is recursive. Moreover the existence of $g_1$ clearly follows from the construction. \qed

\section{The Case of $\Z$}\label{sec:Z}

Decidability of the MSO theory of structures $M=(\Z,<,\overline{P})$ was studied in particular by  Compton \cite{Compton83},
 Sem\"enov \cite{Semenov83,Semenov84}, and Perrin and Schupp \cite{PerrinS86} (see also \cite[chapter 9]{PerrinPin04}). These
  works put in evidence the link between decidability of $MSO(M)$ and computability
  of occurrences and repetitions of finite factors in the word $w(M)$.
  Let us state some notations and definitions (see e.g. \cite[chapter 9]{PerrinPin04}) . A set $X$ of finite words over a finite alphabet $\Sigma$ is
   said to be regular if it is recognizable by some finite automaton. The length of a finite word $u$ is denoted by $|u|$.
   
    Given a $\Z-$word $w$ and a finite word $u$, both over the alphabet $\Sigma$,
     we say that $u$ occurs in $w$ if $w=w_1uw_2$ for some words $w_1$ and $w_2$. We say
     that $w$ is {\em recurrent} if for every regular language $X$ of finite words over $\Sigma$, either no element of
      $X$ occurs in $w$, or in every prefix and every suffix of $w$ there is an occurrence of some element of $X$.

      In particular in a recurrent word $w$,
every finite word $u$ either has no occurrence in $w$, or occurs in every prefix and every suffix of $w$. We say that $w$ is {\em rich} if every
finite word occurs in every prefix and every suffix of  $w$. Given $M=(\Z,<,\overline{P})$, we say that $M$ is \emph{recurrent} if $w(M)$ is.

We have the following result.

\begin{thm}(\cite{Semenov83,PerrinS86}) \label{thm:semZ}
Given $M=(\Z,<,P_1,\dots, P_n)$,
\begin{enumerate}[\em(1)]
\item If $M$ is not recurrent, then every  $c \in \Z$ is definable in $M$.
\item If $M$ is recurrent, then no element is definable in $M$,  and $MSO(M)$ is computable relative to an oracle which,  given any regular language $X$ of finite words over $\Sigma_n=\{0,1\}^n$, tells whether some element of $X$ occurs in $w(M)$.
\end{enumerate}
\end{thm}

In this section we prove the following result.

\begin{prop}\label{prop:Z}
Let $M=(\Z,<,P_1,\dots,P_n)$. There exists an expansion $M'$ of $M$ by some unary predicate $P_{n+1}$
such that $P_{n+1}$ is not definable in $M$, and $MSO(M')$ is recursive in $MSO(M)$. In particular, if $MSO(M)$ is decidable, then $MSO(M')$ is decidable.
\end{prop}

Our proof of Proposition \ref{prop:Z} is based on a definition of $M'$ which depends on whether $M$ is recurrent or not. Proposition \ref{prop:non-recurrent1} deals with  the non-recurrent case, and Proposition \ref{prop:recurrent1} with the recurrent case. These two propositions yield immediately Proposition \ref{prop:Z}. 

\begin{rem}\label{rem:unifnonrec}
Let us discuss uniformity issues related to  Proposition \ref{prop:2ndConstructionN} and Proposition \ref{prop:Z}.
 Proposition \ref{prop:2ndConstructionN} implies that there is an  algorithm which reduces  $MSO(M')$ to $MSO(M)$. This reduction algorithm
   is independent of $M$; it only uses an  oracle for  $MSO(M)$. 
Proposition \ref{prop:Z} implies  a  weaker property. It can be shown indeed that the property for $M$ to be recurrent is not expressible with a MSO sentence in $M$. As a consequence, the algorithm which reduces $MSO(M')$ to $MSO(M)$ depends on $M$. 
 
\end{rem}

\subsection{Non-recurrent case}

We first deal the case when $M=(\Z,<,P_1,\dots, P_n)$ is not recurrent.

\begin{prop}[Expansion of non-recurrent structures]\label{prop:non-recurrent1}
There are  two recursive functions $g_1,g_2$ such that if
$M=(\Z,<,P_1,\dots, P_n)$  is not recurrent, and $c\in \Z$ is definable in $M$
by a formula of quantifier depth $m$, then there exists a function $E_c$ which 
 maps
$M$  to an expansion $M'$   by a predicate $P_{n+1}$ such that
\begin{enumerate}[\em(1)]
\item $P_{n+1}$ is not definable in $M$;
\item $g_1$ computes $T^{k}(M')$ from ${k}$ and $T^{g_2(k+m)}(M)$.
\end{enumerate}
Hence $MSO(M')$ is recursive in $MSO(M)$. In particular, if $MSO(M)$
is decidable, then   $MSO(M')$ is decidable.
\end{prop}

\begin{proof}
Let $c \in \Z$, and let $M_1$ be defined as  $M_1= M_{]-\infty,c[}$ and $M_2$ be defined as $M_{[c,\infty[}$. Then $M=M_1+M_2$.
 Let $M'_1 $ be the expansion of
$M_1$ by the empty predicate $P_{n+1}$ and let $M'_2$ be obtained by
applying the construction of Proposition \ref{prop:2ndConstructionN} to
 $M_2$. Let $M'=M'_1+M'_2$.

Note that the above construction of $M'$ from $M$ depends on $c$. We
denote by $E_c$ the function described above that maps every
 $M=(\Z,<,P_1,\dots, P_n)$ to its  expansion $M'$ by $P_{n+1}$.

By definition, $P_{n+1} \cap [c,\infty[$ is not definable in $M_2$, thus  $P_{n+1}$ is not definable in $M$ by Proposition \ref{prop:def-interval}. Hence
 $M'$ is a non-trivial expansion of $M$.

By Theorem \ref{thm:semZ},  $c$ is definable
  in $M$. Hence, $M_1$ and $M_2$ can be interpreted in $M$, which yields that
  MSO($M_1)$ and  MSO($M_2)$ are recursive in  MSO($M)$. Therefore,
  
    MSO($M'_1)$ and  MSO($M'_2)$ are recursive in  MSO($M)$.
    Finally, applying Theorem \ref{thm:simplecomposition}(a) we
    obtain that MSO($M')$ is  recursive in  MSO($M)$.
\end{proof}

\subsection{Recurrent case}

Now we consider the case when $M=(\Z,<,\overline{P})$ is recurrent. Let us explain why the construction of the previous sub-section cannot be used in this case. Consider a recurrent structure $M$ and let $M'=E_c(M)$ for some $c\in \Z$. We claim that it is possible that  $MSO(M')$ is not recursive in $MSO(M)$.
Indeed, using ideas from \cite{BC08} we can prove
that there exists a recurrent structure $M$ over $\Z$ such that $MSO(M)$ is decidable,
and  $MSO(M_{[c',\infty[})$ is undecidable  for every $c'\in \Z$. Now let $c'$ be the minimal element of $P_{n+1}$.
Observe that $c'$ is definable in $M'$ and therefore, $ M_{[c',\infty[}$
can be interpreted in $M'$. Since, $MSO(M_{[c',\infty[})$ is undecidable, we derive that
    $MSO(M')$
is undecidable.
Hence, $E_c$ does not preserve decidability of recurrent structures. Thus we need a different construction for the recurrent case.

To describe our  construction we introduce first some notations. For every word $w$ over the alphabet $\Sigma_{n+1}=\{0,1\}^{n+1}$ which is indexed by some linear ordering $(A,<)$ we denote by $\pi(w)$ the word $w'$ indexed by $A$ and over the alphabet $\Sigma_{n}=\{0,1\}^{n}$, which is obtained from $w$ by projection over the $n$ first components of each symbol in $w$. The definition and notation extend to $\pi(X)$ where $X$ is any set of words over the alphabet $\Sigma_{n+1}$. Given $M=(\Z,<,\overline{P})$ where $\overline P$ is an $n-$tuple of sets, and any expansion $M'$ of $M$ by a predicate $P_{n+1}$, by definition $w(M)$ and $w(M')$ are words over $\Sigma_n$ and $\Sigma_{n+1}$, respectively, and we have $\pi(w(M'))=w(M)$.

\begin{lem} \label{lem:recur}
 If $M=(\Z,<,P_1,\dots, P_n)$ is recurrent,
 then there is an expansion $M'$ of $M$  by a predicate $P_{n+1} $
 which has the following property:
 \begin{description}
 \item[(*)]
for every $u \in \Sigma_n^*$, if $u $
 occurs in every prefix and every suffix of  $w(M)$, then the same holds in $w(M')$ for every word $u'  \in \Sigma_{n+1}^*$ such that $\pi(u')=u$.
\end{description}
\end{lem}

\begin{proof}
The proof is similar to the proof of Proposition 2.8 in \cite{BC08},
  which roughly shows how to deal with the case when $w(M)$ is rich.

Let $X$ be the set of nonempty words $u \in \Sigma_{n+1}^*$ such that $\pi(u)$ occurs in (every prefix and every suffix of) $w(M)$. We define the expansion $M'$ of $M$ by $P_{n+1}$ in such a way that every element of $X$ occurs in every prefix and every suffix of $w(M')$. This can be done as follows.

 We first define the restriction of $P_{n+1}$ to the interval $[0,\infty[$. Let $(u_i)_{i \in \N}$ be any sequence of elements of $\Sigma_{n+1}^*$ such that $X=\{u_i: i \geq 0\}$, and for every $u \in X$ the set $\{i : u_i=u\}$ is infinite. We shall define sequences of integers $(a_m)_{m \in \N}$ and $(b_m)_{m \in \N}$ and the restriction of $P_{n+1}$ to $[0,b_{m}]$ in such a way that each $u_m$ occurs in $w(M'_{[0,b_{m}]})$, and $a_m$ corresponds to some position in $[0,b_{m}]$ at which $u_m$ starts. We proceed by induction over $m$.
 \begin{enumerate}[$\bullet$]
\item Case $m=0$: we have $u_0 \in X$, therefore the word $\pi(u_0)$ has an occurrence in $w(M_{[0,\infty[})$. Let $a_0 \geq 0$ be the least non-negative integer such that $\pi(u_0)=w(M_{[a_0,a_0+|u_0|[})$, and let $b_{0}=a_0+|u_0|-1$. We define the restriction of $P_{n+1}$ to $[a_0,a_0+|u_0|[$ such that $w(M'_{[a_0,a_0+|u_0|[})=u_0$. Moreover we set $P_{n+1} \cap [0,a_0[ = \emptyset$. 
\item Induction step: let $m \geq 1$. We have $u_m \in X$, therefore the word $\pi(u_m)$ has an occurrence in $w(M_{[b_{m-1},\infty[})$. We define $a_m$ as the least integer greater than $b_{m-1}$  such that $\pi(u_m)=w(M_{[a_m,a_m+|u_m|[})$. We set $b_{m}=a_m+|u_m|-1$. We then define the restriction of $P_{n+1}$ to $[a_m,a_m+|u_m|[$ in such a way that $w(M'_{[a_m,a_m+|u_m|[})=u_i$. Moreover we set $P_{n+1} \cap ]b_{m-1},a_m[ = \emptyset$. 
\end{enumerate}
This construction, and the definition of the sequence $(u_i)$, guarantee that every word $u \in X$ occurs in every suffix of $w(M')$.

We can proceed in a similar way for the definition of the restriction of  $P_{n+1}$ to the interval $]-\infty,0[$, in such a way that every word $u \in X$ occurs in every prefix of $w(M')$. 
\end{proof}

\begin{prop}[Expansion of recurrent structures]\label{prop:recurrent1}
There are  two recursive function $g_1,g_2$ such that if
$M=(\Z,<,P_1,\dots, P_n)$  is   recurrent and  $M'$  is an expansion of $M$ which has property  (*),
then
\begin{enumerate}[\em(1)]
\item $P_{n+1}$ is not definable in $M$;
\item $g_1$ computes $T^{k}(M')$ from ${k}$ and $T^{g_2(k)}(M)$.
\end{enumerate}
Hence $MSO(M')$ is recursive in $MSO(M)$. In particular, if $MSO(M)$
is decidable,  then   $MSO(M')$ is decidable.
\end{prop}

\begin{rem}\label{rem:unifrec}
 Proposition  \ref{prop:recurrent1} implies that there is an  algorithm which reduces  $MSO(M')$ to $MSO(M)$. This reduction algorithm (like the algorithm from  Proposition \ref{prop:2ndConstructionN})
   is independent of $M$; it only uses an  oracle for  $MSO(M)$.
\end{rem}

The proof of Proposition \ref{prop:recurrent1} relies on a refinement of Theorem \ref{thm:semZ}. The analysis of Sem\"enov's proof \cite{Semenov83} of the latter shows that for a given $k \geq 0$, the $k$-type of $MSO(M)$ can actually be computed as soon as we can decide, for {finitely many} regular languages $L$ (which can be computed from $k$), whether some element of $L$ occurs in $w(M)$. More precisely we have the following.

\begin{prop}{(Sem\"enov \cite{Semenov83}, Perrin-Schupp \cite{PerrinS86})}\label{prop:reductionZ}
There exists a recursive function which maps every $k \geq 0$ to a finite sequence $S_k=(L_{k,0},\dots,L_{k,a_k})$ of regular languages of finite words over $\Sigma_n$ such that for every structure $M=(\Z,<,{P_1,\dots,P_n})$ with $w(M)$ recurrent, $T^k(M)$ is computable relative to the set 
$$I_k=\{i : \textrm{ some element of $L_{k,i}$ occurs in $w(M)$}\}.$$
\end{prop}

\noindent{\bf Proof of Proposition \ref{prop:recurrent1}.}

Let $M=(\Z,<,{P_1,\dots,P_n})$ be recurrent, and let $M'$  be an expansion of $M$ which has property  (*).

First let us prove that $P_{n+1}$ is not definable in $M$. For simplicity we use the notations $w,w'$ for $w(M),w(M')$, respectively.  For a fixed $k \geq 0$ there exist finitely many $(k,\{<,P_1,\dots,P_n\})$-Hintikka sentences, say $m_{k}$. Let $u$ be any finite word over $\Sigma_n$ such that $u$ is recurrent in $w$
and $2^{|u|} > (m_{k})^2$. Since $u$ is recurrent in $w$, it follows from the construction of $w'$ that every word $u'$ over $\Sigma_{n+1}$ such that $\pi(u')=u$ is also recurrent in $w'$. There are $2^{|u|}$ distinct such words, say $u'_1,\dots,u'_{2^{|u|}}$. 
For $i=1,\dots,2^{|u|}$, let $(s_i,t_i)$ be such that $w'([s_i,t_i[)=u'_i$ (which implies $w([s_i,t_i[)=u$). By definition of $u$ we have $2^{|u|}  > (m_{k})^2$, which implies that there exist $i<j$ such that $T^k(M_{[0,s_i[})=T^k(M_{[0,s_j[})$ and $T^k(M_{[t_i,\infty[})=T^k(M_{[t_j,\infty[})$. We have  $w([s_i,t_i[)=w([s_j,t_j[)=u$, hence the $T^k(M_{[s_i,t_i[})=T^k(M_{[s_j,t_j[})$. It follows from $w'([s_i,t_i[) \ne w'([s_j,t_j[)$ that there exists $l \in \{0,t_i-s_i-1\}$ such that $w'(s_i+l) \ne w'(s_j+l)$ (where $w'(x)$ denotes the letter at position $x$ in $w'$). Assume that $l \not\in \{0,t_i-s_i-1\}$ (the case $l \in \{0,t_i-s_i-1\}$) is similar).

 We have $w([s_i,s_i+l[)=w([s_j,s_j+l[)$ thus 
$T^k(M_{[s_i,s_i+l[})=T^k(M_{[s_j,s_j+l[})$. It follows that 
\begin{eqnarray}\label{eq:leftside}
T^k(M_{[0,s_i+l[}) & = & T^k(M_{[0,s_i[})+T^k(M_{[s_i,s_i+l[}) \nonumber \\
& = & T^k(M_{[0,s_j[})+T^k(M_{[s_j,s_j+l[}) \nonumber \\
& = & T^k(M_{[0,s_j+l[}). 
\end{eqnarray}
Moreover $w([s_i+l,t_i[)=w([s_j+l,t_j[)$ thus $T^k(M_{[s_i+l,t_i[})=T^k(M_{[s_j+l,t_j[})$. Hence
\begin{eqnarray}\label{eq:rightside}
T^k(M_{[s_i+l,\infty[})& = &T^k(M_{[s_i+l,t_i[})+T^k(M_{[t_i,\infty[}) \nonumber \\
& = & T^k(M_{[s_j+l,t_j[})+T^k(M_{[t_j,\infty[}) \nonumber \\
& = & T^k(M_{[s_j+l,\infty[}) 
\end{eqnarray}
If $k \geq 1$, it follows from (\ref{eq:leftside}), 
 (\ref{eq:rightside}) and Corollary \ref{cor:def-elements} that the elements $s_i+l$ and $s_{j}+l$ cannot be distinguished by any  formula $\varphi(x)$ with quantifier depth less than $k$. Since $w'(s_i+l) \ne w'(s_j+l)$ this implies that $P_{n+1}$ is not definable by any formula of quantifier depth less than $k$. Since this holds for every $k$, we conclude that $P_{n+1}$ is not definable in $M$.

 Let us now prove that one can compute $T^k(M')$ from $k$ and $T^{g_2(k)}(M)$, for some recursive function $g_2$. Let $k \geq 1$. 
Using Proposition \ref{prop:reductionZ} we can compute $S_k=(L'_{k,0},\dots,L'_{k,a_k})$ such that $T^k(M')$ is computable relative to the set 
$$I_k=\{i : \textrm{ some element of $L'_{k,i}$ occurs in $w(M')$}\}.$$ 
Observe that $P_{n+1}$ satisfies $(*)$, thus for every regular language $X' \subseteq \Sigma_{n+1}^*$, some element of $X$ occurs in $w(M')$ iff some element of $\pi(X)$ occurs in $w(M')$. It follows that if we set $L_{k,i}=\pi(L'_{k,i})$ for every $i$, then we have 
$$I_k=\{i : \textrm{ some element of $L_{k,i}$ occurs in $w(M)$}\}.$$
Every language $L'_{k,i}$ is regular, and thus the same holds for $L_{k,i}$. Given any $L_{k,i}$, we can compute a sentence $F_i$ such that $M \models F_i$ iff some element of $L_{k,i}$ occurs in $w(M)$. It follows that $I_k$ is recursive in $MSO(M)$.  
Moreover the computation involves only deciding sentences of the form $F_i$, i.e. uses an oracle for $T^{g_2(k)}(M)$ where $g_2(k)$ is the maximal quantifier depth of sentences of the form $F_i$.
This completes the proof. \qed

\section{Main Result}\label{sec:general}

The next theorem is our main result.

\begin{thm}\label{thm:main}
Let $M=(A,<,P_1,\dots,P_n)$ where $(A,<)$ contains an interval of type $\omega$ or $-\omega$. There exists an expansion $M'$ of $M$ by a unary predicate  $P_{n+1}$
such that $P_{n+1}$ is not definable in $M$, and $MSO(M')$ is recursive in $MSO(M)$. In particular, if $MSO(M)$ is decidable, then $MSO(M')$ is decidable.
\end{thm}
As an immediate consequence we obtain the following
corollary.
\begin{cor}\label{cor:main}
Let $M=(A,<,P_1,\dots,P_n)$ where $(A,<)$ is an infinite scattered linear ordering. There exists an expansion $M'$ of $M$ by some
unary predicate $P_{n+1}$ not definable in $M$ such that $MSO(M')$ is recursive in $MSO(M)$.
\end{cor}

We shall use Shelah's composition method \cite[Theorem 2.4]{Shelah75} (see also \cite{Gurevich85,Thomas97})  which allows us to reduce
the MSO theory of a sum of chains  to the MSO theories of the summands and the MSO theory of the index structure.

\begin{thm}[Composition Theorem \cite{Shelah75}]\label{compositiontheorem}

There exists a recursive function $f$ and an algorithm which, given $k ,l\in \N$, computes the $k$-type of  any  sum  $M=\sum_{i \in I} M_i$  of  chains
over a  signature $\{<,P_1,\dots, P_l\}$
 from the $f(k,l)$-type of the structure
$$(I,<^I,Q_1,\dots,Q_p)$$ where
$$Q_j=\{i \in I: T^{k}(M_i)=\tau_j\} \ \ \ \ j=1,\dots,p$$
and $\tau_1,\dots,\tau_p$  is the list of all $(k,\Delta)$-Hintikka sentences with $\Delta=\{<,P_1,\dots, P_l\}$.
\end{thm}

\noindent{\bf Proof of Theorem \ref{thm:main}. } 

Let $M=(A,<,\overline{P})$ where $(A,<)$ contains an interval of type $\omega$ or $-\omega$. 
The main idea is to use a definable equivalence relation on $A$ to cut $A$ into intervals which are either finite, or have  order type $\omega$, $-\omega$ or $\zeta$, and then  use the previous constructions to define $P_{n+1}$ in these intervals.

Consider the binary relation defined on $A$ by $x \approx y$ iff $[x,y]$ is finite. The relation $\approx$ is a convex equivalence relation, and is definable in $M$.
 
  Let $I$ be the  linear order of the  $\approx$-equivalence classes  for  $(A,<)$.
  Then $M=\sum_{i \in I} M_{A_i}$
  where the $A_i$'s correspond to equivalence classes of $\approx$. Using the definition of $\approx$ and our assumption on $A$, it is easy to check that the $A_i$'s are either finite, or of order type $\omega$, or $- \omega$, or $\zeta$, and that not all $A_i$'s are finite.

We define the interpretation of the new predicate $P_{n+1}$ in every interval $A_i$. The definition proceeds as follows:
\begin{enumerate}[(1)]
\item if some $A_i$ has order type $\omega$ or $-\omega$, then we consider two subcases:

\begin{enumerate}[(a)]
\item if some $A_i$ has order type $\omega$, then we apply to each substructure $M_{A_i}$ of order type $\omega$
the construction given in Proposition \ref{prop:2ndConstructionN}, and add no element of $P_{n+1}$ elsewhere.
\item if no $A_i$ has order type $\omega$, we proceed in a similar way with each substructure $M_{A_i}$ of order type $-\omega$,
 but using the dual  of Proposition \ref{prop:2ndConstructionN} for $ -\omega$.
\end{enumerate}

\item\label{itcaseZ} if no $A_i$ has order type $\omega$ or $-\omega$, then at least one $\approx-$equivalence class $A_i$ has order type $\zeta$. We consider two subcases:
\begin{enumerate}
\item if all $A_i$ with order type $\zeta$ are such that $M_{A_i}$ is recurrent, then we apply to each substructure $M_{A_i}$ of
 order type $\zeta$ the construction given  in Proposition \ref{prop:recurrent1}, and add no element of $P_{n+1}$ elsewhere.
\item otherwise there exists some $A_i$ with order type $\zeta$ and such that $M_{A_i}$ is not recurrent.
Let $\varphi(x)$ be a formula with minimal quantifier depth such that
 $\varphi(x)$ defines an element in some $M_{A_i}$ where $A_i$ has order type $\zeta$.
  For every $M_{A_i}$ such that $A_i$ has order type $\zeta$ and $\varphi(x)$ defines an element $c_i$ in $M_{A_i}$,
   we apply the construction  $E_{c_i}$  from Proposition \ref{prop:non-recurrent1}  to $M_{A_i}$, 
   and add no element of $P_{n+1}$ elsewhere.
\end{enumerate}
\end{enumerate}

\begin{lem}\label{lem:undefG}
$P_{n+1}$ is not definable in $M$.
 \end{lem}

\begin{proof}
This is easy to deduce from the construction: by Proposition \ref{prop:def-interval}, if $P_{n+1}$ was definable in $M$ then the same would hold for $P_{n+1} \cap A_i$ in every sub-structure $M_{A_i}$. 

If we are in case $(1)$, that is, if some $A_i$ has order type $\omega$ or $-\omega$, then we applied the construction of Proposition \ref{prop:2ndConstructionN} to at least one substructure $M_{A_j}$, therefore $P_{n+1} \cap A_j$  is not definable in $M_{A_j}$, which leads to a contradiction. Case $2(b)$ is similar. For case $2(a)$ the contradiction arises from Proposition \ref{prop:non-recurrent1} instead of Proposition \ref{prop:2ndConstructionN}. 
\end{proof}

\begin{lem}\label{lem:reducG}
Let $M'$ be the expansion of $M$ by the predicate $P_{n+1}$. Then $MSO(M')$ is recursive in $MSO(M)$.
 \end{lem}

\begin{proof}

For every $i \in I$ we denote by $M'_i$ (respectively $M_i$) the substructure of $M'$ (resp. $M$) with domain $A_i$. Let $k \geq 1$, and let us compute $T^k(M')$. By Theorem \ref{compositiontheorem}, $T^k(M')$ can be computed from the $f(k,n+1)-$type of the structure $(I,<,Q'_1,\dots,Q'_p)$ where
$$Q'_j=\{i \in I: T^{k}(M'_{i})=\tau'_j\} \ \ \ \ j=1,\dots,p$$
and $\tau'_1,\dots,\tau'_p$ denote the 
list of all $(k,\Delta')$-Hintikka sentences with $\Delta'=\{<,{\overline P},P_{n+1}\}$.

Let $l=f(k,n+1)$. Let us prove that for every $i \in I$, the $l-$type of $M'_i$ can be
computed from the $g(l)-$type of $M_i$ for some recursive function
$g$. Note that $g$ depends on $M$, namely whether we used case $(1)$ or $(2)$ to construct $M'$.

Assume first that we are in case $(1)(a)$.  In this case we first check whether the order type of
the domain $A_i$ of $M_i$ has order type $\omega$; that is, we check whether $M_i \models F$ where $F$ expresses that the domain of the structure  has order type $\omega$. If the answer is positive, then it means that $M'_{i}$ was defined by using the construction given in the proof of Proposition
\ref{prop:2ndConstructionN}, and thus it follows from item $(2)$ in
the latter proposition that $T^l(M'_i)$ is computable from $T^{g_2(l)}(M_i)$. Otherwise we have set $P_{n+1} \cap
A_i=\emptyset$, and in this case the $l$-type of $M'_i$ can be
computed directly from the $l-$type of $M_i$. 

Case $(1)(b)$ is similar to the previous one.

For Case $(2)(a)$, we first check whether  the order type of
the domain $A_i$ is $\zeta$. If the answer is positive, then it means that $M'_{i}$ was defined by using the construction given in the proof of Proposition \ref{prop:recurrent1}, and thus the reasoning is similar to the one for case $(1)(a)$.

For Case $(2)(b)$, we first check whether  the order type of
the domain $A_i$ is $\zeta$. If the answer is positive, then we also check whether the formula $\varphi(x)$ (which was used in the description of case 2(b)) defines an element in $M_{i}$. Namely, we test whether the sentence $\psi \equiv \exists x (\varphi(x) \wedge \forall y \ (\varphi(y) \rightarrow y=x))$ holds in $M_i$. Note that $qd(\psi)=l+2$.
 If the answer is positive again, then it means that $M'_{i}$ was defined by using the construction given in the proof of Proposition
\ref{prop:non-recurrent1}, which implies that the $l$-type of $M'_i$ is computable from the $r-$type of $M_i$, with  $r=\max(l+2,g_2(l+k))$.
Otherwise we have set $P_{n+1} \cap A_i=\emptyset$, and in this case the $l$-type of $M'_i$ can be
computed directly from the $l-$type of $M_i$.

Since for every $i \in I$, $T^{l}(M'_i)$ is computable from $T^{g(l)}(M_j)$ for some recursive function $g$, every formula $Q'_i(x)$ is equivalent to a boolean combination of predicates of the form $Q_j(x)$ where
$$Q_j=\{i \in I: T^{g(l)}(M_{i})=\tau_j\} \ \ \ \ j=1,\dots,r$$
and
$\tau_1,\dots,\tau_r$ denote the sequence of $(g(l),\Delta)-$Hintikka sentences with $\Delta=\{<,{\overline P}\}$.

It follows that $(I,<,Q'_1,\dots,Q'_p)$ is interpretable in  the structure $N=(I,<,Q_1,\dots,Q_r)$, and Lemma \ref{lem:interp} yields that  $MSO(I,<,Q'_1,\dots,Q'_p)$ is recursive in $MSO(N)$. Now the relation $\approx$ is a convex equivalence relation and is definable in $M$. Thus by Lemma \ref{lem:decomp}, the structure $N$ is interpretable in $M$, and $MSO(N)$ is recursive in $MSO(M)$ by Lemma \ref{lem:interp}.
\end{proof}

Theorem \ref{thm:main} follows from Lemmas \ref{lem:undefG} and \ref{lem:reducG}.

\begin{rem}\label{rem:final}
Let us discuss uniformity issues related to Theorem \ref{thm:main}.
\begin{enumerate}[$\bullet$]

\item The choice to expand ``uniformly" all $\approx-$equivalence classes is crucial for the reduction from $MSO(M')$ to $MSO(M)$.
For example, if some $A_i$ has order type $\omega$ and we choose to expand only one such $A_i$
then $MSO(M')$ might become undecidable. This is the case for the structure $M$ considered in \cite{BesCegielski:2009:JMS} (Definition 2.4), which has decidable MSO theory, and is such that the MSO (and even FO) theory of any expansion of $M$ by a constant is undecidable. For this structure all $A_i$'s have order type $\omega$. If we consider the structure $M'$ obtained from $M$ by an expansion of only one $A_i$, then $P_{n+1}$ has a least element, which is definable in $M'$, thus MSO$(M')$ is undecidable.

\item The definition of $P_{n+1}$ in case (\ref{itcaseZ}) depends on whether all components
$A_i$ with order type $\zeta$ are such that $w(M_{A_i})$ is
recurrent, which is not a MSO definable property. Thus the
reduction algorithm from $MSO(M')$ to $MSO(M)$
depends on $M$. 
\end{enumerate}
\end{rem}

\section{Further Results and Open Questions}

Let us mention some possible extensions and related open questions.

First of all, most of our results can be easily extended to the case when the signature contains infinitely many unary predicates.

Our results can be extended to the Weak
MSO logic. In the case $M$
is countable this follows from Soprunov result \cite{Soprunov88}. However, our construction works for
labelled orderings of  arbitrary cardinality.

 An interesting issue is to prove uniform versions of our results in the sense of items $(2)$ in Propositions \ref{prop:2ndConstructionN} and \ref{prop:recurrent1}. A first step would be to generalize Proposition \ref{prop:recurrent1} to all structures $(\Z,<,\overline{P})$.

One can also ask whether the results of the present paper hold for FO logic.
Let us emphasize some difficulties which arise when one tries to adapt the main arguments.
 A FO version of Theorem \ref{th:rechomog} (about the recursive homogeneous set) was already proven in \cite{RabinovichT06}.
  Moreover, using ideas from \cite{Semenov83} one can also give a characterization of structures $M=(\Z,<,\overline{P})$ with a decidable FO theory,
   in terms of occurrences and repetitions of finite words in $w(M)$.
   This allows us to give a FO version of our non-maximality results for labelled orders  over $\omega$ or $\zeta$. However for the general case of  $(A,<,\overline{P})$,
   two problems arise: (1)   the constructions for $\N$ and $\Z$ cannot be applied directly since they are not uniform, and (2)
   the equivalence relation $\approx$ used in the proof of Theorem \ref{thm:main} to cut $A$ into small intervals is not FO definable. We currently investigate these issues.

Finally, we also study the case of labelled   linear orderings
$(A,<,\overline{P})$ which do not contain intervals of order types
$\omega$ or $-\omega$.
      In this case the construction presented in Sect. \ref{sec:general} does not work since the restriction
      of $P_{n+1}$ to each $A_i$ will be empty, i.e.,  our new relation is actually empty.
       In a forthcoming paper we show that it is possible to overcome this issue for countable orders,  and prove that no infinite countable structure $(A,<,\overline{P})$ is maximal.

\subsubsection*{Acknowledgment}

We thank the anonymous referees for useful suggestions.

\bibliographystyle{plain}
\bibliography{max}

\end{document}